\documentclass[11pt,letterpaper,twoside]{article}
\usepackage[centertags]{amsmath}
\usepackage{graphicx,indentfirst,amsmath,amsfonts,amssymb,amsthm,newlfont}
\font\Goth=yinitas scaled \magstep0

\newcommand{\Gth}[1]{\lower2mm\hbox{\Goth #1}}

\usepackage[latin1]{inputenc}
\def\de{\delta}

\def\l1{{\lambda}_1}

\newcommand{\f}{\frac}

\def\x1{{\xi }_{xx}}
\def\x2{{\xi }_{yy}}
\def\x3{{\xi }_{xy}}

\def\e1{{\eta }_{xx}}
\def\e2{{\eta }_{yy}}
\def\e3{{\eta }_{xy}}

\newcommand{\ds}{\displaystyle }

\newtheorem{theorem}{Theorem}

\newtheorem{corollary}{Corollary}
\newcommand{\beqn}{\begin{eqnarray*}}
\newcommand{\eeqn}{\end{eqnarray*}}
\newcommand{\beqnn}{\begin{eqnarray}}
\newcommand{\eeqnn}{\end{eqnarray}}
\newcommand{\p}{\partial}
\newcommand{\bb}{\begin{equation}}
\newcommand{\ee}{\end{equation}}
\newcommand{\ba}{\begin{array}}
\newcommand{\ea}{\end{array}}

\begin{document}
\pagenumbering{arabic}
\title{\huge \bf Note on self-adjoint sub-classes of fourth-order evolution equations with time dependent coefficients}
\author{\rm \large Igor Leite Freire \\
\\
\it Centro de Matemática, Computação e Cognição\\ \it Universidade Federal do ABC - UFABC\\ \it 
Rua Catequese, $242$,
Jardim,
$09090-400$\\\it Santo André, SP - Brasil\\
\rm E-mail: igor.freire@ufabc.edu.br\\}
\date{\ }
\maketitle
\vspace{1cm}
\begin{abstract}
The self-adjoint sub-classes of nonlinear evolution equations of fourth-order with time dependent coefficients are determined, generalizing some recent results. Using the new conservation theorem recently proved by Nail Ibragimov some conservation laws for time dependent equations are established illustrating the main results.
\end{abstract}
\vskip 1cm
\begin{center}
{2000 AMS Mathematics Classification numbers:\vspace{0.2cm}\\
76M60, 58J70, 35A30, 70G65\vspace{0.2cm} }
\end{center}
\pagenumbering{arabic}
\newpage

\section{Introduction}

Since the equation
\bb\label{1.1}
u_{t}+f(u)u_{xxxx}+g(u)u_{x}u_{xxx}+h(u)u_{xx}^{2}+d(u)u_{x}^{2}u_{xx}-p(u)u_{xx}-q(u)u_{x}^{2}=0
\ee
was considered in \cite{qu}, some works have been dedicated to study the symmetry properties and conservation laws of (\ref{1.1}), see \cite{ib1,i1,ga}. 

Equation (\ref{1.1}) encapsulates a wide list of important equations arising from the mathematical physics and mathematical biology, see \cite{ib1,qu,i1}. 

Self-adjoint sub-classes of (\ref{1.1}) were determined in \cite{ib1}. Following Ibragimov \cite{ib3}, a self-adjoint equation has the remarkable property that the nonlocal, or the non-physical, variable $v$ can be eliminated. Then, if there are conservation laws for the equation and its corresponding adjoint equation, after substituting $v=u$, it is possible to establish conservation laws for the original equation. 

For the definitions of adjoint and self-adjoint equations, the author strongly recomend to the diligent reader to read the Ibragimov's works \cite{ib2,ib3,ib4}.

The results on self-adjointness condition of (\ref{1.1}) was generalized in \cite{i1} to the equation
\bb\label{1.2}
\ba{l}
u_{t}+f(u)u_{xxxx}+g(u)u_{x}u_{xxx}+r(u)u_{xxx}+h(u)u_{xx}^{2}\\
\\
+d(u)u_{x}^{2}u_{xx}+p(u)u_{xx}+q(u)u_{x}^{2}+a(u)u_{x}+b(u)=0.
\ea
\ee

In this letter we show that the self-adjoint conditions of the equation (\ref{1.2}) is quite similar to those found in \cite{ib1} and \cite{i1} to the equations (\ref{1.1}) and (\ref{1.2}), respectively. In fact, the aim of this note is the following statement.

\begin{theorem}\label{teo1}
Equation 
\bb\label{1.3}
\ba{l}
u_{t}+f(t,u)u_{xxxx}+g(t,u)u_{x}u_{xxx}+r(t,u)u_{xxx}+h(t,u)u_{xx}^{2}\\
\\
+d(t,u)u_{x}^{2}u_{xx}+p(t,u)u_{xx}+q(t,u)u_{x}^{2}+a(t,u)u_{x}+b(t,u)=0
\ea
\ee
is self-adjoint if and only if
\bb\label{1.4}
g=h+\f{1}{u}(uf)_{u},\,\,\,\,\,\,d=\f{c_{1}(t)}{u}+\f{1}{u}(uh)_{u},\,\,\,\,\,\,q=\f{1}{u}(up)_{u},\,\,\,\,\,\,r=c_{2}(t)+\f{c_{3}(t)}{u},\,\,\,\,\,\,b=\f{c_{4}(t)}{u},
\ee
where  $c_{1}(t),\cdots,c_{4}(t)$ are arbitrary smooth functions. 
\end{theorem}

\begin{proof}
Let 
\bb\label{1.5}
\ba{l}
F:=u_{t}+f(t,u)u_{xxxx}+g(t,u)u_{x}u_{xxx}+r(t,u)u_{xxx}+h(t,u)u_{xx}^{2}\\
\\
+d(t,u)u_{x}^{2}u_{xx}+p(t,u)u_{xx}+q(t,u)u_{x}^{2}+a(t,u)u_{x}+b(t,u)
\ea
\ee
and
\bb\label{1.6}
\ba{lcl}
F^{\ast}&=&\ds{\f{\de}{\de u}[v(u_{t}+fu_{xxxx}+gu_{x}u_{xxx}+ru_{xxx}+hu_{xx}^{2}+du_{x}^{2}u_{xx}+pu_{xx}+qu_{x}^{2}+au_{x}+b)]}\\
\\
&&\ds{=-D_{t}(v)+D_{x}^{4}(vf)-D_{x}^{3}(vgu_{x}+vr)+D_{x}^{2}(2hvu_{xx}+dvu_{x}^{2}+pv)}\\
\\
&&\ds{-D_{x}(gvu_{xxx}+2dvu_{x}u_{xx}+2qvu_{x}+av)+f_{u}vu_{xxxx}+g_{u}vu_{x}u_{xxx}+r_{u}u_{xxx}+h_{u}vu_{xx}^{2}}\\
\\
&&\ds{+d_{u}u_{x}^{2}u_{xx}+p_{u}vu_{xx}+q_{u}vu_{x}^{2}+va_{u}u_{x}+vb_{u},}
\ea
\ee
where 
$$
\f{\de}{\de u}=\f{\p}{\p u}+\sum_{j=1}^{\infty}(-1)^{j}D_{i_{1}}\cdots D_{i_{j}}\f{\p}{\p u_{i_{1}\cdots i_{j}}},
$$
is the Euler-Lagrange operator.

Equation $F^{\ast}=0$ is called adjoint equation to $F=0$ (see the definition in \cite{ib1,ijnmp,i1,ib2,ib3}). Equation $F=0$ is self-adjoint if and only if 
\bb\label{1.7}
\left.F^{\ast}\right|_{v=u}=\phi F,
\ee
for some differential function $\phi$. Substituting equations (\ref{1.5}) and (\ref{1.6}) into (\ref{1.7}) and comparing the coefficient of $u_{t}$, it is obtained that $\phi=-1$ and, up to differential consequences, from the coefficients of the remaining terms, it is obtained
\bb\label{1.8.1}
(uf)_{u}-ug+uh=0,
\ee
\bb\label{1.8.2}
(up)_{u}-uq=0,
\ee
\bb\label{1.8.3}
3(uf)_{uuu}-3(ug)_{uu}+2(ud)_{u}+(uh)_{uu}=0,
\ee
\bb\label{1.8.4}
(ur)_{uu}=0,
\ee
\bb\label{1.8.5}
(ub)_{u}=0,
\ee
whose solution is (\ref{1.4}).
\end{proof}

In what follows, the prime $'$ means $\f{d}{d u}$. 
\begin{corollary}(Freire \cite{i1})\label{cor1}
Equation $(\ref{1.2})$ is self-adjoint if and only if
$$
g=h+\f{1}{u}(uf)',\,\,\,\,\,\,d=\f{c_{1}}{u}+\f{1}{u}(uh)',\,\,\,\,\,\,q=\f{1}{u}(up)',\,\,\,\,\,\,r=c_{2}+\f{c_{3}}{u},\,\,\,\,\,\,b=\f{c_{4}}{u},
$$
where $c_{1},\cdots,c_{4}$ are arbitrary constants and $f,\,h$ and $p$ are arbitrary functions of $u$.
\end{corollary}

\begin{corollary}(Bruzón, Gandarias, Ibragimov \cite{ib1}, Freire \cite{i1})\label{cor2}
Equation $(\ref{1.1})$ is self-adjoint if and only if
$$
g=h+\f{1}{u}(uf)',\,\,\,\,\,\,d=\f{c_{1}}{u}+\f{1}{u}(uh)',\,\,\,\,\,\,q=\f{1}{u}(up)',
$$
where $c_{1}$ is an arbitrary constant and $f,\,h$ and $p$ are arbitrary functions of $u$.
\end{corollary}

\begin{corollary}(Bruzón, Gandarias, Ibragimov \cite{ib1})\label{cor3} Equation
$$
u_{t}+f(u)u_{xxxx}=0
$$
is self-adjoint if and only
\bb\label{1.9}
u_{t}+\f{a}{u}u_{xxxx}=0,
\ee
where $a=const$.
\end{corollary}

\begin{corollary}(Freire \cite{ijnmp})\label{cor4}
Equation
\bb\label{1.8}u_{t}+a(t,u)u_{x}+b(t,u)=0\ee
is self-adjoint if and only if
$$b=\f{\lambda(t)}{u},$$
for some smooth function $\lambda=\lambda(t)$. 
\end{corollary}

\section{An application}

Here we present an application of conservation laws and self-adjoint equations using the recent new conservation theorem due to Ibragimov \cite{ib2}. 



Consider the self-adjoint equation
\bb\label{2.1}
u_{t}+f(t)uu_{x}+g(t)u_{xxx}=0.
\ee

Let $F$ be a function such that $F'=f$. A Lie point symmetry generator of (\ref{2.1}) is
\bb\label{2.2}
X=F(t)\f{\p}{\p x}+\f{\p}{\p u}.
\ee
From \cite{ib2}, a conserved vector for equation (\ref{2.1}) and its adjoint
\bb\label{2.2'} 
v_{t}+f(t)uv_{x}+g(t)v_{xxx}=0
\ee
is $C=(C^{0},C^{1})$, where
\bb\label{2.3}
\ba{lcl}
C^{0}&=&\ds{\tau {\cal L}+W\,\f{\p {\cal L}}{\p u_{t}}},\\
\\
C^{1}&=&\ds{\xi {\cal L}+W\left[\f{\p {\cal L}}{\p u_{x}}-D_{x}\f{\p {\cal L}}{\p u_{xx}}+D_{x}^{2}\f{\p {\cal L}}{\p u_{xxx}}\right]}\\
\\
&&\ds{+D_{x}(W)\left[\f{\p {\cal L}}{\p u_{xx}}-D_{x}\f{\p {\cal L}}{\p u_{xxx}}\right]+D_{x}^{2}(W)\f{\p {\cal L}}{\p u_{xxx}}},
\ea
\ee
$W=\eta-\tau u_{t}-\xi u_{x}$ and ${\cal L}=v(u_{t}+f(t)uu_{x}+g(t)u_{xxx})$.

Substituting $\xi=F(t),\,\tau=0,\,\eta=1$ and $W=1-F(t)u_{x}$ into (\ref{2.3}), it is obtained
\bb\label{2.3'}
\ba{lcl}
C^{0}&=&v(1-F(t)u_{x}),\\
\\
C^{1}&=&F(t)vu_{t}+f(t)vu+g(t)v_{xx}-F(t)g(t)u_{x}v_{xx}+F(t)g(t)v_{x}u_{xx}.
\ea
\ee

Setting $v=u$ in (\ref{2.3'}), the components of the vector become
$$C^{0}=u-F(t)D_{x}\left(\f{u^{2}}{2}\right),\,\,\,\,\,C^{1}=f(t)u^{2}+F(t)D_{t}\left(\f{u^{2}}{2}\right)+g(t)u_{xx}.$$
Thus
$$D_{t}C^{0}+D_{x}C^{1}=D_{t}(u)-D_{x}\left(f(t)\f{u^{2}}{2}\right)-F(t)D_{t}D_{x}\left(\f{u^{2}}{2}\right)+F(t)D_{t}D_{x}\left(\f{u^{2}}{2}\right)+D_{x}(f(t)u^{2}+g(t)u_{xx})$$
and then, 
\bb\label{c1}
C=(u,f(t)u^{2}/2+g(t)u_{xx})
\ee 
provides a conserved field for equation (\ref{2.1}). In particular 
\bb\label{c2}
C=(u,f(t)u^{2}/2)
\ee
is a conserved vector for the time dependent inviscid Burgers equation 
\bb\label{ibe1}
u_{t}+f(t)uu_{x}=0,
\ee
$D_{t}C^{0}+D_{x}C^{1}=0$, where
\bb\label{c3}
C^{0}=u,\,\,\,\,\,C^{1}=\f{u^{2}}{2},
\ee
provides a conservation law for the inviscid Burgers equation
\bb\label{ibe2}
u_{t}+uu_{x}=0
\ee
and 
\bb\label{c4}
C=(u,u^{2}/2+u_{xx})
\ee
is a conserved vector for the KdV equation
\bb\label{kdv}
u_{t}+uu_{x}+u_{xxx}=0.
\ee

Further examples can be found in \cite{ib1,ijnmp,i1,ib2,ib3}.

\section{Conclusion}

Theorem \ref{teo1} generalizes results on self-adjoint equations previously obtained in \cite{ib1,ijnmp,i1}, namely, corollaries \ref{cor1}, \ref{cor2}, \ref{cor3} and \ref{cor4}. Theorem \ref{teo1} together the new conservation theorem proved by Ibragimov \cite{ib2} can be used for establishing conservation laws for equations (\ref{1.2}), since the coefficients obey (\ref{1.4}). Many applications can be found in \cite{ib1, ijnmp, i1,ib2,ib3,jk3}. 

The conserved field (\ref{c1}) provides an infinite number of conservation laws for equation (\ref{2.1}) parametrized by the functions $f(t)$ and $g(t)$. Particular cases of these conservation laws are well known, see references cited above. 

To the best of the author's knowledgment, the conserved vector (\ref{c3}) for the Korteweg-de Vries equation, using (\ref{2.3}) (or (\ref{2.3'})), was first obtained by Ibragimov in his fundamental work \cite{ib2}, see \cite{i1}. Taking $f(t)=1$ in (\ref{2.1}), the vector $C=(u,u^{2}/2+g(t)u_{xx})$ is a conserved vector for the time dependent KdV equation
$$u_{t}+uu_{x}+g(t)u_{xxx}=0.$$
Thus $Div(C)=0$ provides conservation laws for equations listed in the Case 1 of \cite{jk2}. Hence these conservation laws complement the results obtained in \cite{jk3}.

With regard to the inviscid Burgers equation (\ref{ibe2}), the conserved vector (\ref{c3}) is a particular case of the conserved fields obtained in \cite{ijnmp}. However, up to the author's knowledgment, the present note is the first work that the general conserved vector (\ref{c1}) is found using (\ref{2.3}), as well as its particular vector (\ref{c2}), wich provides a conservation law for (\ref{ibe1}).

\end{document}